\documentclass[journal,12pt,draftclsnofoot,onecolumn]{IEEEtran}
\usepackage{amsmath}
\usepackage{graphicx}
\usepackage{epstopdf,color}
\usepackage{amssymb}
\usepackage{amsthm}
 \usepackage[center]{caption}
\usepackage{tikz}
\usetikzlibrary{patterns}
\usepackage{ cite}
\include{notation}
\begin{document}

\title{ Spectrum Allocation for ICIC Based Picocell}
\author{Suman Kumar \hspace*{.5in} Sheetal Kalyani \hspace*{.5in} K. Giridhar  \\
\hspace{0in} Dept. of Electrical Engineering \\
{\tt \{ee10d040,skalyani,giri\}@ee.iitm.ac.in}\\
}
\maketitle
\begin{abstract}
In this work, we analytically  study the impact of spectrum allocation scheme in picocells on the coverage probability (CP) of the Pico User (PU), when the macro base stations (MBSs) employ either fractional frequency reuse (FFR) or soft frequency reuse (SFR).  Assuming a  fixed size for the picocell, the CP expression is derived for a PU present in   either  a FFR or SFR based deployment, and when the PU uses either the centre or the edge frequency resources.  Based on these expressions, we propose two possible frequency allocation schemes for the picocell when FFR is employed by the macrocell. The CP and the average rate expressions for both these schemes are derived, and it is  shown that these schemes outperform the conventional scheme where no inter-cell interference coordination (ICIC) is assumed. The impact of both schemes on the macro-user  performance is also analysed. When SFR is used by the MBS, it is shown that the CP is maximized when the PU uses  the same frequency resources as used by the centre region. 
\end{abstract}
\section{Introduction}
Picocells are low-powered operator-deployed base stations, which are primarily used to improve the coverage of hot spots and cell edge \cite{1184120, 1308951}.  They enhance  the frequency reuse in the system, and thereby increasing the available per-user bandwidth. 
However,  picocell deployment and planning can pose many challenges. The high cost of licensed spectrum may force the operator to allocate the same frequency to the low-power picocell and the macrocell, leading to significant co-channel interference in such heterogeneous networks. For homogeneous networks (comprising of only macrocells),  inter-cell interference coordination (ICIC) strategies such as FFR and SFR have been  proposed to mitigate the impact of co-channel interference. Recently, both FFR and SFR schemes have been included in fourth generation wireless standards such as WiMAX  802.16m and 3GPP-LTE  release 8 \cite{5534591}. Our focus here is on deciding what frequency resources should be allocated to the picocell when the macrocell employs ICIC schemes (FFR and SFR).

Spectrum allocation in such tier networks  has been extensively studied in literature. A fast and effective power control algorithm  and a suboptimal allocation algorithm are proposed in \cite{6542983} where PU and MU both share the same spectrum. A pico location optimization method is presented in \cite{6362538} for spectrum allocation strategies. A new cell selection method based on the resource specific Signal-to-Interference-plus-noise-Ratio (SINR) value is discussed in \cite{6092863} where both PU and MU  share the same spectrum. An ICIC scheme based on avoiding the primary interfering source is described in \cite{6398987}.  All these contributions relate to the spectrum allocation to the picocell, but do not consider an ICIC scheme (FFR or SFR) for the macrocell. A FFRopa (FFR with frequency occupation ordering and power adaptation) scheme for the macrocell with various frequency allocation schemes for the picocell is proposed in \cite{6666396} where through simulation it is shown that this method outperforms FFR. Spectrum allocation in the femtocell is studied in \cite{6476610, suman, 5734795}

Heterogeneous networks have been extensively studied in \cite{6171996,6220221, 6171998, 6171997, robert} and references therein using analytically tractable models.  Picocells are mostly deployed to improve the CP of hot spots and/or cell edge. Hence, the network operator will be interested to know the improvement in the performance for a given area  which could be a hot spot or a cell edge location. Therefore, the model and the analysis in    \cite{robert}  is  well suited for the picocell analysis, and we assume the same model.

The contribution of this work is to analytically  derive the CP and rate of the PU for two specific FFR schemes applicable for these Hetnets. In order to do that, we first distinguish five different frequency resources usage scenarios in the picocell, namely: $(a)$ Centre frequency resources of FFR scheme employed in macrocell,  $(b)$ Edge frequency resources of FFR scheme, $(c)$ Neighbouring edge frequency resources of FFR scheme,  $(d)$ Centre frequency resources of SFR scheme, and  $(e)$ Edge frequency resources of SFR scheme. These different frequency resources are shown in Fig. \ref{fig:fig0}. We propose two frequency allocation schemes for the picocell when FFR is employed in the macrocell, where the PUs are segregated  into categories, based on Signal-to-Interference-Ratio (SIR) threshold $S_{tp}$ namely: cell-centre PUs and cell-edge PUs. The proposed scheme $1$ allocates centre frequency resources of FFR scheme to the cell-centre PUs and neighbouring edge frequency resources of FFR scheme to the cell-edge PUs, while scheme $2$ allocates edge frequency resources  to the cell-centre PUs, and the neighbouring edge frequency resources of FFR scheme to the cell-edge PUs. We show that the  CP for both the schemes is higher than the conventional scheme\footnote{ In the conventional scheme, no ICIC is assumed,  and the macrocell and picocell  use the same spectrum simultaneously, i.e.,  unity frequency reuse is used in both the regions.}. Also, depending on the SIR threshold $S_{tm}$ of macrocell, either Scheme $1$ or Scheme $2$ may be preferred to deliver a higher average rate when compared to the conventional scheme. Further, we  analyse the impact of proposed schemes on the CP of MU. We also show that in SFR deployment, the CP is maximized when the PU uses the frequency resources of the centre.
\begin{figure}[ht]
\centering
\begin{tikzpicture}
\draw[fill=white!10] [line width=.01cm](2,2) circle (1);
\draw[fill=white!10] [line width=.01cm](2,2) circle (0.5);
\path[draw] (1.85,1.9) -- (2.15,1.9)--(2,2.2)--cycle;
\node at (1.8,1.8){\scriptsize $F_0$};
\draw[fill=white!10] [line width=.01cm](2,4) circle (1);
\draw[fill=white!10] [line width=.01cm](2,4) circle (0.5);
\node at (1.8,3.8){\scriptsize $F_0$};
\path[draw] (1.85,3.9) -- (2.15,3.9)--(2,4.2)--cycle;
\draw[fill=white!10] [line width=.01cm](3.75,3) circle (1);
\draw[fill=white!10] [line width=.01cm](3.75,3) circle (0.5);
\node at (3.55,2.8){\scriptsize $F_0$};
\path[draw] (3.6,2.9) -- (3.85,2.9)--(3.75,3.2)--cycle;
\node at (2.75,2){\scriptsize $F_1$};
\node at (2.75,4){\scriptsize $F_2$};
\node at (4.5,3){\scriptsize $F_3$};
\draw[->] (2.75,4.1)--(5,5) node[pos=0.8,above] {\scriptsize Neighbouring Edge frequency resources};
\draw[->] (4.5,3.1)--(5.1,5) node[pos=1,above] {};
\draw[->] (2.75,1.9)--(3.5,1.5) node[pos=1,below] {\scriptsize  Edge frequency resources};
\draw[->] (1.7,1.95)--(0.5,3) node[pos=1,above] {\scriptsize  Center frequency resources};
\draw[fill=white!10] [line width=.01cm](7,3) circle (1.5);
\draw[fill=white!10] [line width=.01cm](7,3) circle (0.75);
\path[draw] (6.85,2.9) -- (7.15,2.9)--(7,3.15)--cycle;
\node at (6.8,3.3){\scriptsize $F_2 \& F_3$};
\node at (8,3){\scriptsize $F_1$};
\draw[->] (8,2.9)--(9,2) node[pos=1,below] {\scriptsize  Edge frequency resources};
\draw[->] (7,3.45)--(8,5) node[pos=1,above] {\scriptsize  Center frequency resources};
\draw[->] (2,4)--(1,5) node[pos=0.9,above] {\scriptsize MBS};
\node at (2,6){ FFR};
\node at (7,6){SFR};
\draw[->] (0.5,0.5)--(1.3,1.3) node[pos=0,below] {\scriptsize Reference Cell};
\end{tikzpicture}
\caption{Frequency resources allocation in macrocell in FFR and SFR}
 \label{fig:fig0}
        \end{figure}
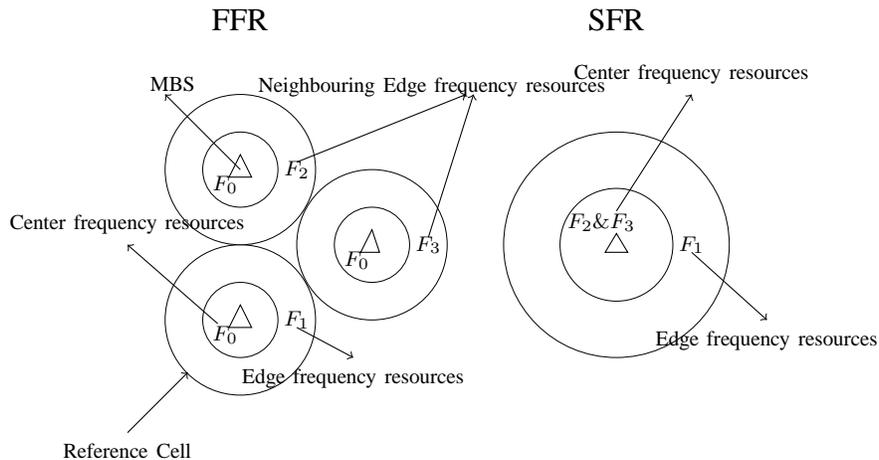
\section{Downlink System Model}
We consider a hybrid model similar to the one proposed in \cite{robert}, where a picocell with a fixed network area, i.e., a circle with radius $R$, is considered.  The locations of MBSs and other pico base-stations (PBSs) outside of the fixed network area are modelled according to a Poison Point Process (PPP) \cite{robert}. Similar to \cite{robert}, a guard region of radius $R_g$ from the cell-edge is imposed around the fixed cell so that no interfering base-stations (BSs) are assumed to be inside this guard region.
Since we are primarily interested in the performance of a typical PU,  the following model for spatial location of nodes is considered. A fixed picocell at the origin is assumed, and  the locations of the MBSs are modelled by a spatial PPP on $B(o,R+R_g)^c = \mathbb{R}^2\setminus B(o,R+R_g)$ of density $\lambda_m$, and is denoted by $\phi$.  The PBSs are modelled by another  PPP of density $\lambda_p$ on $B(o,R+R_g)^c$  and is denoted by $\psi$. We also assume that $\phi$ is independent of $\psi$.

A standard path loss model $x^{-\alpha}$ with $\alpha\geq 2$ is used. Channel fading power between any two nodes is assumed to be exponential, and is independent across nodes. The MBSs and PBSs  transmit   at a fixed power of $P_m$ and $P_p$,  respectively. The network is assumed to be interference limited, and hence noise power is taken to be zero. We  assume that PUs will be associated with a PBS, if they are within the  network area of the desired PBS. All the users outside the PBS network  area are assumed to be MUs and they associate to the nearest MBS. Hence, the network area of the macrocell consist of Voronoi regions.
The SIR of the PU, which is at distance $r$ from the associated PBS may be written as:
\begin{equation} 
\eta_p(r)=\frac{P_p r^{-\alpha}{g_p}}{I_{\phi}+ I_{\psi\setminus P_0}}, I_{\phi}=\sum\limits_{i\in \phi}P_m{d_i}^{-\alpha}{h_i} \text{ and } I_{\psi\setminus P_0}=\sum\limits_{j\in (\psi \setminus P_0)}P_p{r_j}^{-\alpha}{g_j}
\label{eq:two}
\end{equation}
where $\phi$ denotes the set of all MBS, and $\psi\setminus P_0$ denotes the set of all PBS except the serving PBS ($P_0$). The distance from the  $j^{th}$  interfering PBS and $i^{th}$ interfering  MBS to the typical PU are denoted by  $r_j$ and $d_i$, respectively. The channel fading power from the $j^{th}$ interfering PBS and the $i^{th}$ interfering  MBS are denoted by $g_j$ and $h_i$, respectively, and $g_p$ denotes the fading gain from the typical PU to PBS. Note that $g_p$, $g_i$ and $h_i$ are all independent and identically exponentially distributed with unit mean, i.e, $\exp(1)$. Similarly,  the SIR of the MU, which is at distance $r$ from the associated MBS may be written as:
\begin{equation} 
\eta_m(r)=\frac{P_m r^{-\alpha}{g_m}}{I_{\phi\setminus M_0}+I_{\psi}},
\label{eq:three}
\end{equation}
where $g_m$ denotes the fading gain from the typical MU to MBS, which is exponentially distributed with unit mean. Here $I_{\phi\setminus M_0}$ denotes the interfering power from all the MBSs except the nearest MBS, and $I_{\psi}$ denotes the interfering power from all the PBSs.

Both FFR and SFR   classify users into cell-centre MUs and cell-edge MUs, based on a SIR threshold denoted by $S_{tm}$. Users having SIR\footnote{Here, the SIR in the context of an OFDM system would typically be the wideband SIR averaged over all sub carriers, in the current frame. Overtime, as the channel and the pathloss parameters change, it is possible that a cell-edge user gets reclassified as a cell-centre user, or vice-versa.} higher than $S_{tm}$ are classified as cell-centre MUs, and otherwise they are denoted as cell-edge users. FFR uses frequency reuse $\frac{1}{\delta}$ for the cell-edge MUs to boost up the SIR while providing unity frequency reuse to the cell-centre MUs. In other words, FFR need a total of $\delta+1$ sub-bands. One sub-band is used for the cell-centre users and one among $\{1,\cdots, \delta \}$ is chosen with equal probability which is used for the cell-edge users. However, SFR employs a reuse factor of $\frac{1}{\delta}$ on the cell-edge, and utilizes the neighbouring frequency resources at the cell-centre as shown in Fig. \ref{fig:fig0}. SFR also uses $\beta>1$ times higher transmit power for cell-edge MUs when compared to cell-centre MUs to enhance the performance of cell edge users. 
\section{Coverage Probability}
We consider a fixed picocell at the origin and compute the probability that a user within its cell is in coverage, i.e.,  the SIR of the user is greater than  $T$. In other words, CP for the PU is  defined as  $ P_c = \mathbb{P}[\text{SIR}>T].$
To find out the CP for the PU, we need to know the probability density function (pdf) of $r$, 
where $r$ is the distance between typical
PU to the tagged PBS. Assuming that the user distribution is uniform, the probability that a PU will be at distance $r$, given a radius $R$ of the picocell is given by $\frac{2r}{R^2}$.
Therefore, the pdf of $r$ is,
\begin{equation}
f_R(r)= \left\{
\begin{array}{rl}
&\frac{2r}{R^2}, r\leqslant R\\
&0,  r>R.
\end{array} \right.
\end{equation}
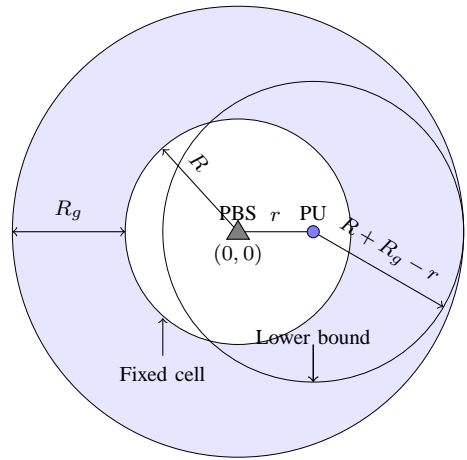
\begin{figure}[ht]
\centering
\begin{tikzpicture}
\draw[fill=blue!10] [line width=.01cm](5,5) circle (3);
\draw[fill=white!10] [line width=.01cm](5,5) circle (1.5);
\draw [line width=.01cm](6,5) circle (2);
\draw[-] (5,5)--(6,5) node[pos=.5,sloped,above] {\scriptsize$r$};
\node at (5,5.25){\scriptsize PBS};
\node at (6,5.25){\scriptsize PU};
\draw[->] (6,5)--(7.73,4) node[pos=.5,sloped,above] {\scriptsize$R+R_g-r$};
\draw[->] (5,5)--(4,6.1) node[pos=.7,sloped,above] {\scriptsize$R$};
\draw[<->] (2,5)--(3.5,5) node[pos=.5,sloped,above] {\scriptsize$R_g$};
\node at (6,3.6){\scriptsize Lower bound};
\draw [black,line width=.02  cm] [->](6,3.5)--(6,3);
\draw[->](4,3.35)--(4,3.85) node[pos=0,below] {\scriptsize Fixed cell};
\node at (5,4.7) {\scriptsize $(0,0)$};
\path[draw, fill=black!50] (4.85,4.9) -- (5.15,4.9)--(5,5.15)--cycle;
\draw[fill=blue!50] [line width=.01cm](6,5) circle (0.08);
\end{tikzpicture}
\caption{Network coverage area of the PBS that is considered for the lower bound. The circle centred around the PBS with radius $R$ is the original network area.}
\label{fig:fig00}
        \end{figure}
The distance between the considered PU (which is a distance $r$ from the PBS) and the interfering base stations can vary from $R+R_g-r$ to $R+R_g+r$. Deriving the exact expression  is difficult since the user is not exactly centred at the origin \cite{robert}. Therefore, similar to \cite{robert}, we consider a picocell region of radius $R+R_g-r$, i.e., a ball of radius $R+R_g-r$ around a  PU, that gives an upper bound on the interference
power as shown in Fig. \ref{fig:fig00}. The suitable  length for the guard region is discussed in \cite{robert} where it is shown that guard region radius corresponding to the picocell should be $R_g^{(p)}=R$ and that corresponding to the macrocell be $R_g^{(m)}=(\frac{P_m}{P_p})^\frac{1}{\alpha}R$. It has been assumed in \cite{robert} that interference has two components: one interferer at the boundary of the guard region, and another interferers outside the guard region. We assume a sparse density for the picocell, and also the coverage radius of the picocell is considered to be  small\footnote{In heterogeneous network, considered for 3GPP-LTE, picocell can use a range extension factor (bias in cell association SIR) which can greatly increase the coverage radius of the PBS. In our work we do not assume any such range expansion scheme at the PBS.}. In other words, we assume that $\lambda_m<\lambda_p\ll \frac{1}{\pi (R)^2}$. Intuitively, in such condition, there is very little chance that a  dominant interferer to be located at the boundary of the guard  region.  Hence, we do not consider any interferer at the boundary of the guard edge region. We also observed  that without considering an interferer at the boundary of the guard region, the analytical result matches better with the simulation results.

 We now derive the CP of a typical PU when the macrocell employs FFR for the above deployment model.
\newtheorem{theorem}{Lemma}
 \begin{theorem}
The CP of a typical PU, using the centre frequency resources band $F_0$ of the FFR  in an interference limited scenario is given by 
\begin{equation}
P_{f,c}(T)= \int_{r>0}^{R}\exp(-\pi r^2 T^{2/\alpha}K(r)) \frac{2r}{R^2}\mathrm{d}r,
\label{eq:four}
\end{equation}
where
\begin{equation}
 K(r)=\left(\frac{P_m}{P_p}\right)^{2/\alpha}C(L_m(r),\alpha)\lambda_m + C(L_p(r),\alpha)\lambda_p.
\label{eq:lemma4}
\end{equation}
\begin{equation*}
C(L_p(r),\alpha)=\frac{{}_2 F_1(1,\frac{\alpha-2}{\alpha},2-\frac{2}{\alpha},-\{L_p(r)\}^\alpha)}{\alpha-2}  \{L_p(r)\}^{2-\alpha} \text{ }\mathrm{ and } \text{ } L_p(r)= \frac{R+R_g^{(p)}-r}{r}(T)^{-\frac{1}{\alpha}}.
\label{hyper1}
\end{equation*}
\end{theorem}
\begin{proof}
 The proof is provided in the Appendix.
\end{proof}
Here, the first term in $K(r)$ is due the macrocell and second term is due to the interfering picocell. The expression derived in Lemma $1$ is equivalent to the  CP of PU in a conventional scheme in which reuse one PBS and reuse one MBS both use the same frequency resource. Now, we consider the case when edge frequency resources of macrocell are used by PU. The CP of a typical PU, using the edge frequency resources band $F_1$ of the FFR follows from Lemma $1$ with a slight modification in the  interference due to macrocell. Since FFR uses  frequency reuse $\frac{1}{\delta}$ for the cell edge users (i.e., it reuses a frequency sub-band from $1, \cdots, \delta$ with equal probability), the MBS interferers density  will be ``thinned version'' of the original MBS interferers density seen by the PU \cite{6042301}. In other words, the MBSs interferer density will be $\frac{\lambda_m}{\delta}$ instead of $\lambda_m$.
Hence, the  CP of a typical PU, using the edge frequency resources band $F_1$ of the FFR in an interference limited scenario is given by
\begin{equation}
P_{f,e}(T)= \int_{r>0}^{R}\exp(-\pi r^2 T^{2/\alpha}\hat{K}(r)) \frac{2r}{R^2}\mathrm{d}r,
\label{eq:theorem1} 
\end{equation}
where 
\begin{equation}
\hat{K}(r)=\left(\frac{P_m}{P_p}\right)^{2/\alpha}C(L_m(r),\alpha)\frac{\lambda_m}{\delta} +C(L_f(r),\alpha)\lambda_p.
\label{eq:ffr}
\end{equation}
Now we will derive the CP of PU when  MBS uses FFR and PUs use the edge frequency resources of neighbouring cells. Here we assume that picocell  connects to the nearest MBS via the back-haul. The distance between PU to the nearest macrocell is $q$.  Using null probability, the pdf of $q$ can be written as \cite{6042301} 
\begin{equation}
f_q(Q)=2\pi\lambda q e^{-\lambda \pi (q^2-(R_g^{(m)}+R-r)^2)},\text{ } \mathrm{ where } \text{ } q\geq R_g^{(m)}+R-r
\end{equation}
 \begin{theorem}
The  CP of a typical PU, using the edge frequency resources band $F_2$ and $F_3$ of the FFR (edge frequency resources band of the neighbouring macrocell) in an interference limited scenario is given by
\begin{equation}
\textstyle
P_{f,\hat{e}}(T)=\int_{r>0}^{R}\exp\{-\pi r^2 (T )^{2/\alpha}C(L_p(r),\alpha)\lambda_p \}\mathbb{E}_q\left[\exp\{-\pi r^2 \left(\frac{TP_m}{P_p}\right)^{2/\alpha}C(L_m(r,q),\alpha)\frac{\lambda_m}{\delta} \}\right]\frac{2r}{R^2}\mathrm{d}r\label{eq:eight}
\end{equation}
\end{theorem}
\begin{proof}
The proof is similar to that outlined for Lemma $1$, except for the fact that there will be no interference from the nearest MBSs, and the density of MBSs interference will be thinned by $\delta$.
The Laplace transform of interference due to MBS  when PU uses the neighbouring edge frequency resources ($ \mathcal{L}_{\hat{I}_{\phi\setminus M_0}}(s)$) needs to be derived,  and is given by
 \begin{eqnarray}
   && \mathcal{L}_{\hat{I}_{\phi\setminus M_0}}(s)=\mathbb{E}_{\phi,h_i}[\exp(-s\sum_{i \in \phi \slash M_0} h_i d_i^{-\alpha})] \nonumber\\
&=&\left[\exp\left(-2\pi\frac{\lambda_m}{\delta}\int_{q}^\infty \frac{sx^{-\alpha}}{1+sx^{-\alpha}}x\text{d}x\right )\right]\label{lower_limit}.
\end{eqnarray}
The lower limit of  the integration in \eqref{lower_limit} is $q$ due to the fact that all the interfering MBSs are at least a distance greater than $q$. Using a change of variable $t=s^{-\frac{1}{\alpha}}x$, $ \mathcal{L}_{\hat{I}_{\phi\setminus M_0}}(s)$ can now be simplified as
\begin{eqnarray}  
 \mathcal{L}_{\hat{I}_{\phi\setminus M_0}}(s)&=&\left[\exp\left(-2\pi \frac{\lambda_m}{\delta} s^{\frac{2}{\alpha}}\int_{L_m(q,r)}^\infty \frac{t}{1+t^{\alpha}}\text{d}t\right)\right]\nonumber
 \end{eqnarray}
where $L_m(q,r)= \frac{q}{r}\left(T \frac{P_m}{P_p}\right)^{-\frac{1}{\alpha}}$. Again, $ \mathcal{L}_{\hat{I}_{\phi\setminus M_0}}(s)$ can be further simplified as (see Appendix), 
\begin{eqnarray}  
 \mathcal{L}_{\hat{I}_{\phi\setminus M_0}}(s)     = \exp\{-\pi s^{2/\alpha}C(L_c(q,r),\alpha)\frac{\lambda_m}{\delta} \} \label{eq:seven}
 \end{eqnarray}
$\text{where }C(L_c(q,r),\alpha)=  \frac{{}_2 F_1(1,\frac{\alpha-2}{\alpha},2-\frac{2}{\alpha},-L_c(q,r)^\alpha)}{\alpha-2}L_c(q,r)^{2-\alpha}\nonumber$. 
\end{proof}

We now start with a discussion on the impact of FFR scheme on the CP of the PU where we consider $P_m=46$dBm, $P_p=30$dBm, $\lambda_m=0.385/Km^2$, $\lambda_p=1.155/Km^2$ $\delta=3$, $\alpha=4$, and $R=200$ meters. Fig. \ref{fig:fig1} shows the CP of PU when MBS uses FFR deployment. When the centre frequency resources of FFR are used in the picocell (this is the  conventional scheme) then it gives the lowest CP due to the fact that PU experiences  interference from all the MBSs. Using the edge frequency resources give a higher coverage than the centre frequency resource because of frequency reuse $\frac{1}{\delta}$. However, using the  neighbouring edge frequency resource gives the maximum CP although it uses frequency reuse $\frac{1}{\delta}$ due to the fact that there is no interference coming from the geographically nearest MBS.
\begin{figure}[ht]
 \centering
 \includegraphics[scale=0.32]{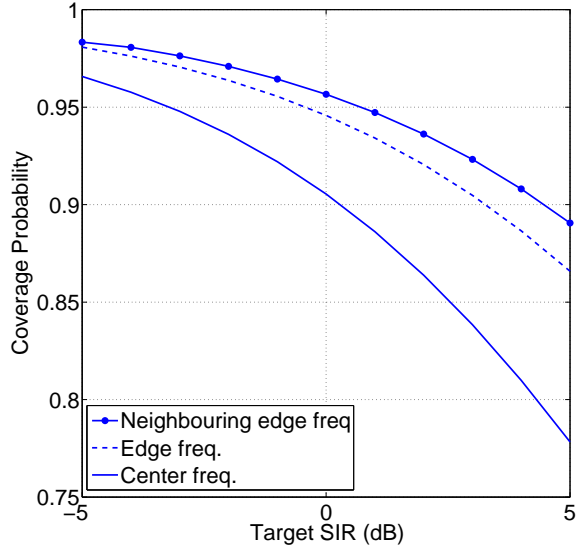}
 \caption{CP of pico user for FFR deployment in macrocell.}
 \label{fig:fig1}
\end{figure}
We would now like to analyse the impact of various FR allocation schemes in the picocell. In the next section, we propose two scheme for FR allocation and show their impact on  the CP of PU and MU.
\section{Proposed Frequency allocation schemes for PU}
We divide the PUs into two parts: cell-centre PUs and cell-edge PUs based on the SIR threshold ($S_{tp}$). The user seeing a  SIR $>S_{tp}$ is considered as cell-centre PU; else, the user is as a cell-edge PUs. 
\subsection{Proposed Scheme 1 for FFR}
In the proposed scheme 1 (PS1), considering the reference cell of FFR deployment in Fig.\ref{fig:fig0} we allocate frequency $F_0$ to the cell-centre PUs and frequency $F_2$ and $F_3$ to the cell-edge PUs. In other words, the centre frequency of macrocell will be used by cell-centre PUs, and neighbouring macro cell-edge frequencies would be used by cell-edge PUs.
Now, we derive the CP of PUs when  PS$1$ is used in the picocell. The  CP of a typical PU, when PU uses PS$1$ is given by
\begin{equation}
\textstyle
P_{ps1}=\int\limits_{0}^{R}\left(P[\eta_p(r)>T|\eta_p(r)>S_{tp}]P[\eta_p(r)>S_{tp}]+P[{\eta}^e_p(r)>T|\eta_p(r)<S_{tp}]P[\eta_p(r)<S_{tp}]\right)\frac{2r}{R^2}\text{d}r
\end{equation}
where ${\eta}^e_p(r)$ denotes the SIR experienced by PU when PU uses neighbouring cell-edge frequencies. Here, the first term  denotes the CP  due to cell-centre PUs, and the second term  denotes the CP due to cell-edge PUs.
Note that 
\begin{equation}
P[\eta_p(r)>T|\eta_p(r)>S_t]P[\eta_p(r)>S_t]\stackrel{(g)}{=}P[\eta_p(r)>\max\{S_{tp},T\}]\label{first_term}
\end{equation}
Here $(g)$ follows from the Bayes' rule. Since fading power is assumed to be independent across the sub-bands,  one obtains, 
\begin{equation}
P[{\eta}^e_p(r)>T|\eta_p(r)<S_{tp}]=P[{\eta}^e_p(r)>T]\label{second_term}
\end{equation}
Using \eqref{first_term} and \eqref{second_term}, $P_{Ps1}$ can be simplified as
\begin{equation}
P_{ps1}=\int\limits_{0}^{R}(P[\eta_p(r)>\max\{S_{tp},T\}]+P[{\eta}^e_p(r)>T]P[\eta_p(r)<S_{tp}])\frac{2r}{R^2}\text{d}r\label{ffr_pico1}
\end{equation}
$P[\eta_p(r)>T]$ has been derived in Lemma $1$, and is given by
\begin{equation}
P[\eta_p(r)>T]=\exp(-\pi r^2 T^{2/\alpha}K(r))\label{ffr_pico2}
\end{equation}
where $K(r)$ is defined in \eqref{eq:lemma4}. Also, $P[{\eta}^e_p(r)>T]$ is derived in Lemma $2$ and is given by
\begin{equation}
\textstyle
P[{\eta}^e_p(r)>T]=\exp\{-\pi r^2 (T )^{2/\alpha}C(L_p(r),\alpha)\lambda_p \}   \mathbb{E}_q\left[\exp\{-\pi r^2 \left(\frac{TP_m}{P_p}\right)^{2/\alpha}C(L_m(r,q),\alpha)\frac{\lambda_m}{\delta} \}\right]\label{ffr_pico3}
\end{equation}
Using \eqref{ffr_pico1}, \eqref{ffr_pico2} and \eqref{ffr_pico3}, $P_{ps1}$ can be evaluated. 
\subsection{Proposed Scheme 2 for FFR}
In the second scheme (PS2), again with reference to the FFR as depicted in Fig. \ref{fig:fig0}, we allocate frequency $F_1$ to the cell-centre PUs and frequency $F_2$ and $F_3$ to the cell-edge PUs. Thus, the edge frequency of macrocell will be used by cell-centre PUs and the neighbouring cell-edge frequencies would be used by cell-edge PUs. The CP of PS$2$ is directly follows from the CP of PS$1$ derived in the previous section except for the fact that instead of $K(r)$ in \eqref{ffr_pico2}, $\hat{K}(r)$ (given in \eqref{eq:ffr}) will be used since the edge frequency of FFR is used by the cell-centre PU. 

Now, our focus will be on the CP of cell-edge MUs when PS$2$ is used in picocell. First we derive the CP of cell-edge MUs when macrocell uses FFR and there is no picocell. Although it is already been  derived in \cite[Theorem 2]{6047548}, it will be seen that our derived expression is simpler to evaluate, and matches with the simulation result. The CP of cell-edge MUs at a distance $r$ from the BS is given by 
\begin{equation}
P_{ed}(r)=P[\hat{\eta}(r)>T|\eta(r)<S_{tm}].
\end{equation}
Here $\eta(r)$ denote the SIR experienced by MU when it uses centre frequency, and $\hat{\eta}(r)$ denotes the SIR experienced by MU when it uses cell-edge frequency and  the picocell is absent. Since fading power is assumed to be independent across the sub-bands,  
\begin{equation}
P[\hat{\eta}(r)>T|\eta(r)<S_{tm}]=P[\hat{\eta}(r)>T]\label{ffr_macro2}
\end{equation}
The CP of a typical cell-edge MU is then given by
\begin{equation}
P_{ed}=\frac{\int\limits_{0}^{\infty}P[\hat{\eta}(r)>T|\eta(r)<S_{tm}]P[\eta(r)<S_{tm}]f_R(r)\text{d}r}{\int\limits_{0}^{\infty}P[\eta(r)<S_{tm}]f_R(r)\text{d}r}\label{eq:mu}
\end{equation}
where $f_R(r)$ is the pdf of $r$ (which is the distance between user and the nearest MBS), and is given by  \cite{6042301} 
\begin{equation}
f_r(R)=2\pi\lambda re^{-\lambda \pi r^2},  r>0
\end{equation}
Using  \eqref{ffr_macro2}, $P_{ed}$ in \eqref{eq:mu} can be simplified as
\begin{equation}
P_{ed}=\frac{\int\limits_{0}^{\infty}P[\hat{\eta}(r)>T]P[\eta(r)<S_{tm}]f_R(r)\text{d}r}{\int\limits_{0}^{\infty}P[\eta(r)<S_{tm}]f_R(r)\text{d}r}\label{ffr_macro3}
\end{equation}
In \cite{6042301}, it has been shown that 
\begin{equation}
P[\eta(r)>T]=e^{-\pi\lambda r^2\rho(T,1)} \text{ and } P[\hat{\eta}(r)>T]
=e^{-\pi\lambda r^2\rho(T,\delta)}
\label{basic}
\end{equation}
where 
$ $
\begin{equation*}
\textstyle
\rho(T,\delta)=\frac{T^{\frac{2}{\alpha}}}{\delta}\int\limits_{T^{-2/\alpha}}^{\infty}\frac{1}{1+u^{\alpha/2}}\text{d}u
=\frac{2T}{\delta(\alpha-2)}{}_2 F_1(1,\frac{\alpha-2}{\alpha},2-\frac{2}{\alpha},-T)
\end{equation*}
Thus,  $P_{ed}$ can be further simplified to
\begin{equation}
P_{ed}=\frac{ 2\pi\lambda_m\int\limits_{0}^{\infty}re^{-\pi\lambda_m r^2}e^{-\pi\lambda_m r^2\rho(T,\delta)}(1-e^{-\pi\lambda_m r^2\rho( S_{tm},1)})\text{d}r}{ 2\pi\lambda_m\int\limits_{0}^{\infty}re^{-\pi\lambda_m r^2}(1-e^{-\pi\lambda_m r^2\rho( S_{tm},1)})\text{d}r}\label{ffr_macro4}
\end{equation}
Solving the integrals,  $P_{ed}$ can be rewritten as 
\begin{equation}
P_{ed}=\frac{\frac{1}{1+\rho(T,\delta)}-\frac{1}{1+\rho(T,\delta)+\rho( S_{tm},1)}}{1-\frac{1}{1+\rho( S_{tm},1)}}=\frac{1+\rho( S_{tm},1)}{(1+\rho(T, \delta))(1+\rho(T,\delta)+\rho( S_{tm},1))}
\end{equation}
Now, we derive the CP of cell-edge MUs when PS$2$ is used in the picocell, and the macrocell employs FFR. The CP of cell-edge MUs at a distance $r$ from the BS is given by 
\begin{equation}
P_{ed, ps2}(r)=P[\hat{\eta}_m(r)>T|\eta_m(r)<S_{tm}]
\end{equation}
and  $\hat{\eta}_m(r)$ denotes the SIR experienced by MU when it uses FR$\delta$, and the picocell uses the PS$2$. Since fading power is assumed to be independent, we have
\begin{equation}
P[\hat{\eta}_m(r)>T|\eta_m(r)<S_{tm}]=P[\hat{\eta}_m(r)>T]\label{ffr_macro5}
\end{equation}
The CP of a typical cell-edge MU is then given by
\begin{equation}
P_{ed, ps2}=\frac{\int\limits_{0}^{\infty}P[\hat{\eta}_m(r)>T]P[\eta_m(r)<S_{tm}]f_R(r)\text{d}r}{\int\limits_{0}^{\infty}P[\eta_m(r)<S_{tm}]f_R(r)\text{d}r}\label{ffr_macro6}
\end{equation}
Note that
\begin{equation*}
P[\hat{\eta}_m(r)>T]=\mathbb{P}\left[\frac{g_m r^{-\alpha}}{\hat{I}_{\phi\setminus M_0} +\frac{P_p}{P_m}I_\psi}>T\bigg|r\right],
\end{equation*}
\begin{eqnarray}
= \mathbb{P}\left[g_m>Tr^\alpha\left(\hat{I}_{\phi\setminus M_0} + \frac{P_p}{P_m}I_\psi\right)\right]
=\mathcal{L}_{\hat{I}_{\phi\setminus M_0}}(Tr^\alpha )\mathcal{L_{I_\psi}}\left(\frac{P_p}{P_m}Tr^\alpha\right),
\label{eq:lemma7}
\end{eqnarray}
We know from \eqref{basic} that $\mathcal{L}_{\hat{I}_{\phi\setminus M_0}}(s)=P[\hat{\eta}(r)>T]=e^{-\pi\lambda_m r^2\rho(T,\delta)}$, and
\begin{eqnarray}
&&\mathcal{L}_{I_\psi}(s)=\mathbb{E}_{\psi,g_j}\left[\exp\left(-s\sum_{j \in \psi } g_j r_j^{-\alpha}\right)\right] \nonumber\\
&\stackrel{(h)}=&\exp\left(-2\pi\lambda_p\int_{R}^\infty \frac{sx^{-\alpha}}{1+sx^{-\alpha}}x\text{d}x\right )\label{eq:upper1}\\  
&=&\exp\left(-2\pi \lambda_ps^{\frac{2}{\alpha}}\int_{L_m(r)}^\infty \frac{t}{1+t^{\alpha}}\text{d}t\right)\nonumber\\
&=& \exp\{-\pi s^{2/\alpha}C(L_m(r),\alpha)\lambda_p \} \label{eq:lemma6}
\end{eqnarray}
where $L_m(r)= \frac{R}{r}(T)^{-\frac{1}{\alpha}}$.
The lower limit of the  integral in $(h)$ is $R$ because of the fact that all the interfering PBS are at least a distance greater than $R$. Thus, the CP of a typical cell-edge MU is given by
\begin{equation}
P_{ed, ps2}=\frac{ 2\pi\lambda_m\int\limits_{0}^{\infty}e^{-\pi\lambda_m r^2}e^{-\pi\lambda_m r^2\rho(T,\delta)} \exp\{-\pi (\frac{P_p T}{P_m})^{\frac{2}{\alpha}}r^2C(L_m(r),\alpha)\lambda_p \}(1-e^{-\pi\lambda_m r^2\rho( S_{tm},1)})\text{d}r}{1-\frac{1}{1+\rho( S_{tm},1)}}\label{ffr_macro7}
\end{equation}
We now derive the CP of cell-centre MU. The CP of cell-centre MU without picocell is given in \cite[Theorem 3]{6047548}. The CP of cell-centre MU when PS$1$ is used in the picocell is given by
\begin{equation}
P_{cen,ps1}=\frac{\int\limits_{0}^{\infty}P[\eta_m(r)>\max\{S_{tm},T\}]f_R(r)}{\int\limits_{0}^{\infty}P[\eta_m(r)>T]f_R(r)}
\end{equation}
Note that
\begin{equation*}
P[\eta_m(r)>T]=\mathbb{P}\left[\frac{g_m r^{-\alpha}}{ I_{\phi\setminus M_0} +\frac{P_p}{P_m}I_\psi}>T\right]=\mathcal{L}_{I_{\phi\setminus M_0}}(Tr^\alpha )\mathcal{L_{I_\psi}}(\frac{P_p}{P_m}Tr^\alpha),
\label{eq:lemma5}
\end{equation*}
We know from \eqref{eq:lemma6},
$\mathcal{L_{I_\psi}}(\frac{P_p}{P_m}Tr^\alpha)=\exp\{-\pi (\frac{P_p T}{P_m})^{\frac{2}{\alpha}}r^2C(L_p(r),\alpha)\lambda_p \} $ and 
$\mathcal{L}_{I_{\phi\setminus M_0}}(Tr^\alpha )=P[\eta(r)>T]=e^{-\pi\lambda_m r^2\rho(T,1)}$ from \eqref{basic}, and hence CP of cell-centre MU can be derived. The CP of cell-centre MU does not change when PS$2$ is used in the picocell, since PS$2$ does not use centre frequency resources.
\subsection{Average Rate}
In this subsection, we derive average rate for both the schemes. The average  rate $R$ can be written as
\begin{equation*}
R=E[\ln(1+\text{SIR})]=\int\limits_{t>0}P[\ln(1+\text{SIR})>t]\text{d}t.
\end{equation*}
Using the fact that  $\ln(1+\text{SIR})$ is a monotonic increasing function for SIR, one obtains,
\begin{equation}
R=\int\limits_{t>0}P[\text{SIR}>e^t-1]\text{d}t.
\label{eq:rate}
\end{equation}
This is equivalent to computing CP for $T=e^t-1$ and integrating it over $t$. The CP for the PS$1$ is given in \eqref{ffr_pico1}, and thus the average rate of PU in PS$1$ is given by
\begin{equation}
R_{ps1}=\int\limits_{t>0}\left(\int\limits_{0}^{R}(P[\eta_p(r)>\max\{S_{tp},e^t-1\}]+P[\hat{\eta}_p(r)>e^t-1]P[\eta_p(r)<S_{tp}])\frac{2r}{R^2}\text{d}r\right)\text{d}t\label{rate_ps1}
\end{equation}
where $P[\eta_p(r)>T]$ and $P[\hat{\eta}_p(r)>T]$ are given in \eqref{ffr_pico2} and \eqref{ffr_pico3}, respectively. Similarly, the average rate of PU in PS$2$ can be derived. In order to choose $S_{tp}$, we define the sub-bands allocation to the MU and PU.  Here, based on the SIR threshold  sub-bands allocation can be done \cite{kumar2014} \cite{6047548}. In other words,
 \begin{equation}
N_c=P[\eta_m >S_{tm}]N_t \text{ and }N_e=\frac{N_t-N_c}{3}
\end{equation} 
where $N_t=F_0+F_1+F_2+F_3$, $N_c=F_0$, and $N_e=F_1=F_2=F_3$, denote the  total sub-bands, centre sub-bands, and  the edge sub-bands, respectively. Since PS$1$ allocates the centre frequency band of the macrocell to the cell-centre PUs,  $S_{tp}$ for the PS$1$ can be chosen such that
\begin{equation}
P_{f,c}(S_{tp})=\frac{N_c}{N_t}
\end{equation}
Here, $P_{f,c}(S_{tp})$ is derived in Lemma $1$, Similarly, the $S_{tp}$ for the PS$2$ can be chosen such that
\begin{equation}
P_{f,e}(S_{tp})=\frac{N_e}{N_t},
\end{equation}
where $P_{f,e}(S_{tp})$ is given by \eqref{eq:theorem1}. To compare the average rate for both the schemes, we define a normalized average rate which is average rate times the fraction of the total sub-bands used in that particular scheme. In other words, the normalized average rate of PS$1$ is $\frac{R_{ps1}(F_0+F_2+F_3)}{N_t}$, and similarly the normalized average rate of PS$2$ is $\frac{R_{ps2}(F_1+F_2+F_3)}{N_t}$.
\subsection{Analysis and Comparison of PS1 and PS2 for FFR}
In this subsection, we analyse the performance of PS$1$ and PS$2$ and compare both these schemes. Fig. \ref{fig:fig2} shows the CP of PU when picocell uses either PS$1$ or PS$2$, and it also plots the CP of the conventional scheme. Here again we take $P_m=46$dBm, $P_p=30$dBm, $\lambda_m=0.385/Km^2$, $\lambda_p=1.155/Km^2$ $\delta=3$, and $R=200$ meters It can be seen that  PS$2$ provides better CP than PS$1$, and both have a significantly better performance when compared to the conventional scheme.
\begin{figure}[ht]
 \centering
 \includegraphics[scale=0.32]{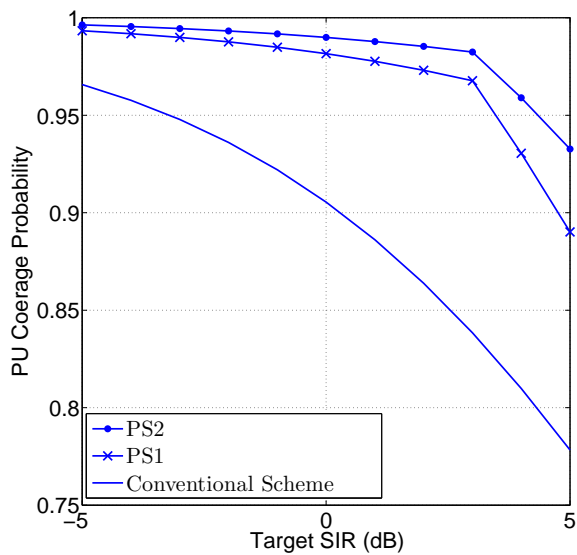}
 \caption{CP of pico user for the proposed schemes. Here $S_{tp}=3$dB and $\alpha=4$.}
 \label{fig:fig2}
\end{figure}

\begin{figure}[ht]
 \centering
 \includegraphics[scale=0.33]{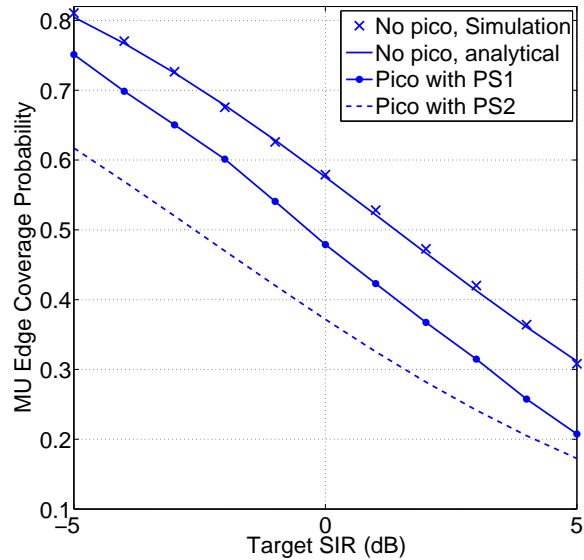}
 \caption{CP of cell-edge MUs where picocell are using different proposed schemes. Here $S_{tm}=0$dB and $\alpha=3.2$.}
 \label{fig:fig3}
\end{figure}
\begin{figure}[ht]
 \centering
 \includegraphics[scale=0.33]{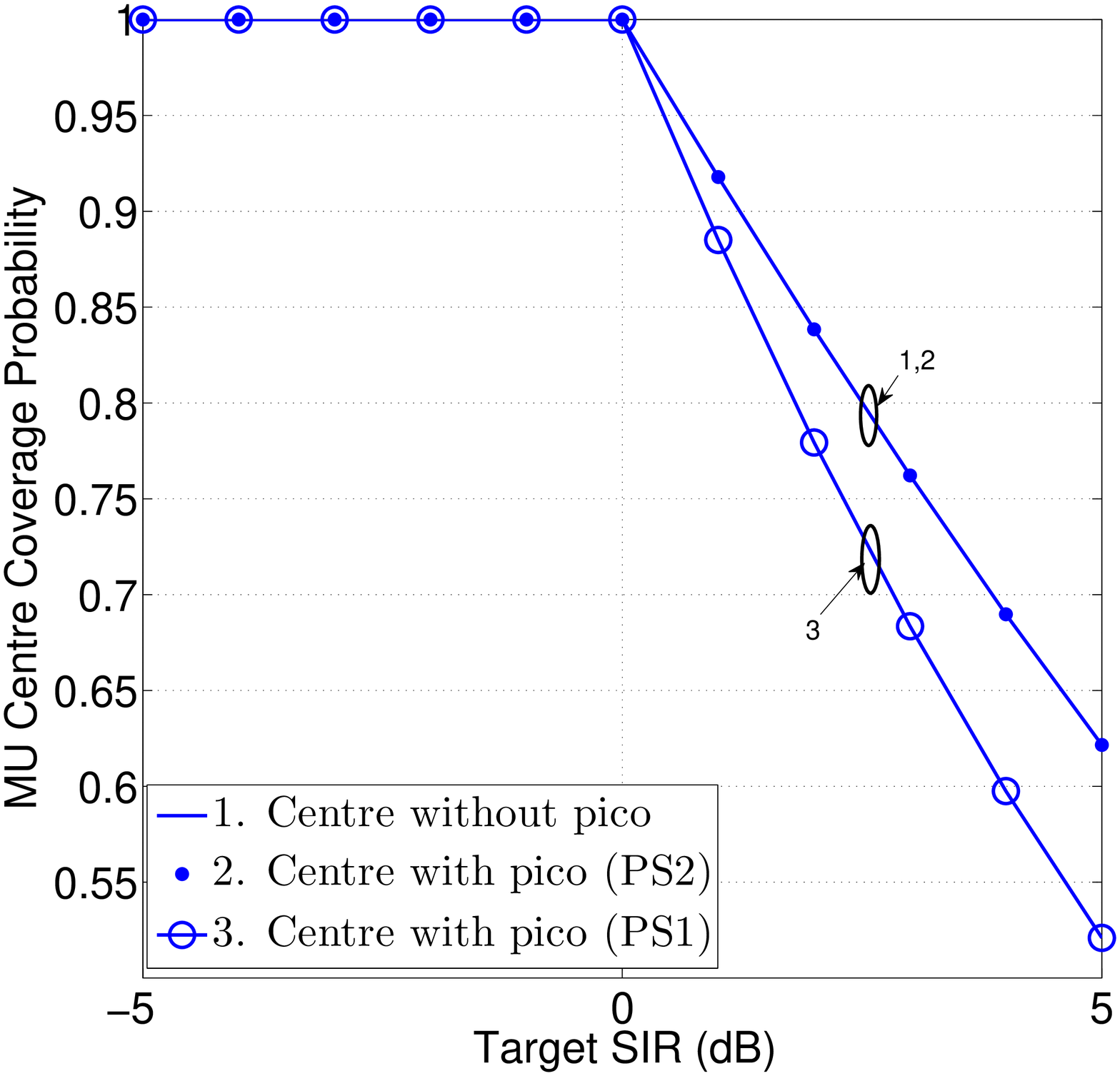}
 \caption{CP of cell-centre MUs where picocell uses different proposed schemes. Here $S_{tm}=0$dB and $\alpha=3.2$.}
 \label{fig:fig5}
\end{figure}
We will now study the impact of the proposed schemes on the CP of the MUs. Fig. \ref{fig:fig4} shows the CP of cell-edge MUs for three cases: $(i)$ when there is no picocell, $(ii)$ when picocell uses PS$1$, and $(iii)$ when picocell use PS$2$. Here, the CP of cell-edge MU when picocell uses PS$1$ is plotted using simulation. It can be observed that the CP of cell-edge MU is lowest when PS$2$ is used in the picocell. The CP of cell-edge MU when there is no picocell is highest followed by when PS$1$ is used in the picocell.  Fig. \ref{fig:fig5} show the impact of proposed schemes on the centre CP. The CP of cell-centre degrades when picocell uses PS$1$. However, there is no impact of PS$2$ on the CP of the cell-centre MU, since PS$2$ does not use cell-centre resources at the picocells. Fig. \ref{fig:fig6} plots the normalized average rate with respect to SIR threshold of macrocell. It can be seen that as $S_{tm}$ increases, centre frequency of macrocell ($N_c$) decreases, i.e., the sub-bands of picocell for PS$1$ decreases and the sub-bands of picocell for PS$2$ increases. Hence, the normalized average rate of PS$1$ decreases and  the normalized average rate of PS$2$ increases.

Based on the above results, we see that in FFR deployments, both PS$1$ and PS$2$ have an advantage over the conventional scheme, but neither of them is uniformly better than the other. Then, the question arises which of these two schemes should the operator use in the picocell? We provide following guidelines that could help this choice:
\begin{enumerate}
\item  Depending on the SIR threshold chosen for the macrocell, PS$1$ or PS$2$ can be selected.
\item  The proposed schemes allocation could also depend on the MUs around the picocell. If cell-edge MUs around one particular picocell are higher then that picocell should use the PS$2$ rather PS$1$ and vice versa. 
\item  It could also depend on the location of picocell with respect to nearest macrocell. It is seen by simulation that as picocell location  moves towards the edge of nearest macrocell, the gap between the CP of PU corresponding to PS$1$ and PS$2$ increases. In other words, the CP of PU corresponding to PS$2$ is much higher than the   CP of PU corresponding to PS$1$ when picocell location is at the edge of macrocell. Hence, depending on the location of picocell, proposed scheme selection can be made.
\end{enumerate}
\begin{figure}[ht]
 \centering
 \includegraphics[scale=0.33]{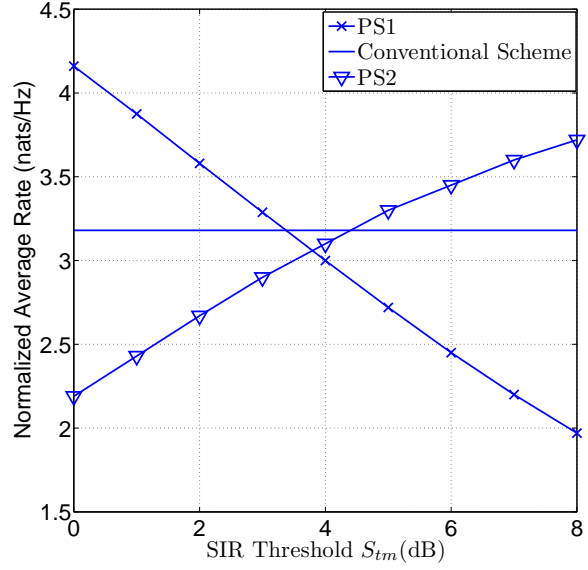}
 \caption{Normalized average rate of PU. Here $\alpha=4$.}
 \label{fig:fig6}
\end{figure}

\section{Frequency Allocation for PU When MBS Employ SFR}
In this section, we first derive the CP for PU when MBS employ SFR. In order to calculate the CP of PU when MBS uses SFR, we divide the interference due to MBSs into two parts: interference due to the nearest MBS,  and the interference due to all the remaining MBSs. The effective interference power due to all remaining MBSs  is denoted by $\eta P_m$, where as in  \cite{6047548}, $\eta=\frac{\delta-1+\beta}{\delta}$. 
 \begin{theorem}
The CP of a typical PU,  using the centre frequency resources band $F_2$ or $F_3$ of the SFR for an interference limited scenario is given by
\begin{equation}
\textstyle
 P_{s,c}(T)=\int_{r>0}^{R}\frac{2r}{R^2}\exp\{-\pi r^2 (T )^{2/\alpha}C(L_p(r),\alpha)\lambda_p \} \text{ } \mathbb{E}_q\left[\frac{P_p}{P_p+Tr^\alpha P_m {q}^{-\alpha}}e^{\{-\pi r^2 \left(\frac{\eta T P_m}{P_p}\right)^{2/\alpha}C(L_c(r),\alpha)\lambda_m} \}\right]\mathrm{d}r \label{eq:SFR}
\end{equation}
\end{theorem}
\begin{proof}
CP of PU can be written as,
\begin{equation*}
 P_{s,c}=\int_{r>0}^{R}\mathbb{P}[\eta_p(r)>T|r]f_r(r)\text{d}r,
\end{equation*}
\begin{equation*}
 =\int_{r>0}^{R}\frac{2r}{R^2} \mathbb{P}\left[\frac{g_p r^{-\alpha}}{\eta \frac{P_m}{P_p}I_{\phi\setminus M_0} +I_{\psi\setminus P_0} + \frac{P_m}{P_p}{q}^{-\alpha}{h_0}}>T|r\right]\text{d}r, 
\end{equation*}
\begin{equation}
=\int_{r>0}^{R}\frac{2r}{R^2}\mathcal{L}_{I_{\phi\setminus M_0}}(T r^\alpha \eta \frac{P_m}{P_p})\mathcal{L}_{I_{\psi\setminus P_0}}(Tr^\alpha) \exp\{-Tr^\alpha \frac{P_m}{P_p} {q}^{-\alpha}{h_0}\}\text{d}r.
\label{eq:six}
\end{equation}
Here the first term $\mathcal{L}_{I_{\phi\setminus M_0}}(T r^\alpha \eta \frac{P_m}{P_p})$ is due to interference from all the MBS except the nearest MBS, and the second term $\mathcal{L}_{I_{\psi\setminus P_0}}(Tr^\alpha)$ is due to interference from all the picocell except the serving picocell, and third term $\exp\{-Tr^\alpha \frac{P_m}{P_p} {q}^{-\alpha}{h_0}\}$ is due to interference from the nearest MBS. Where, $h_0$ is the fading power gain from the nearest MBS and $h_0\sim \exp(1)$ hence one  obtains,
\begin{equation}
\mathbb{E}_{h_0}[\exp\{-Tr^\alpha\frac{P_m}{P_p} {q}^{-\alpha}{h_0}\}]=\frac{P_p}{P_p+Tr^\alpha P_m {q}^{-\alpha}}
\end{equation}
$\mathcal{L}_{I_{\phi\setminus M_0}}(s)$  is given in \eqref{eq:seven} and $\mathcal{L}_{I_{\psi\setminus P_0}}(Tr^\alpha)$ has been derived in the Appendix (please see Eq. \eqref{eq:lemma2}). Hence using \eqref{eq:seven} and \eqref{eq:lemma2}, Eq. \eqref{eq:six} can be simplified to obtain \eqref{eq:SFR}.
\end{proof}
Now, the case when PU uses the edge frequency resources band $F_1$ of the SFR will be analysed. The CP of PU in this case  directly follows from an application of Lemma $3$ except for the fact that the edge frequency uses $\beta$ times higher power than the centre frequency. Therefore, interference due to the  nearest MBS will be $ (1+\beta Tr^\alpha \frac{P_m}{P_p} {q}^{-\alpha})^{-1}$ instead of $ (1+Tr^\alpha \frac{P_m}{P_p} {q}^{-\alpha})^{-1}$. Hence, the  CP of a typical PU,  using the edge frequency resources band $F_1$ of the SFR for an interference limited scenario is given by an expression identical to \eqref{eq:SFR}, but with  $ (1+Tr^\alpha \frac{P_m}{P_p} {q}^{-\alpha})^{-1}$ within the expectation replaced by $ (1+\beta Tr^\alpha \frac{P_m}{P_p} {q}^{-\alpha})^{-1}$.

The impact of SFR on the CP of PU is plotted in Fig. \ref{fig:fig2}. Here again we take $\delta=3$, $\alpha=3.2$, and $R=200$ meters. This result shows that the edge frequency resources give the lowest CP among all the frequency resources in the SFR scheme.  Centre frequency resources provide a higher coverage than the edge frequency resources because of the fact that they use  $\beta$ times less power for the nearest MBS. Now, as we increase  $\beta$, the power on the centre frequency resources  decreases with respect to edge frequency resources, and this results in a higher coverage. Hence, centre frequency resources in SFR should only be used by the picocell.

\begin{figure}[ht]
 \centering
 \includegraphics[scale=0.33]{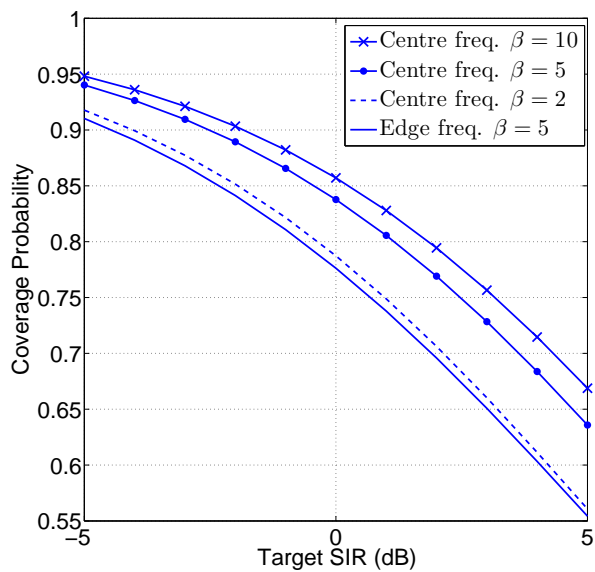}
 \caption{CP of pico user for SFR deployment in macrocell.}
 \label{fig:fig4}
\end{figure}

\section{Conclusions}
This work considered FFR and SFR deployments for the macrocell, and  provided  the CP and rate of a PU by assuming a fixed size picocell. We  proposed two schemes for the picocell when the macrocell employs FFR. The CP and rate of proposed schemes were derived, and it was shown that the proposed schemes outperforms the conventional scheme. The impact of the proposed schemes on the CP of MU was also discussed. Finally, the spectrum allocation in picocell was briefly discussed when the macrocell employed SFR. Future work could consider the effect of range expansion (or picocell strength bias) \cite{6166483}, on the performance of the proposed schemes. Further extension of this work could include the effect of correlation among the interferers \cite{kumar2014analysis}. Another important extension of this work is to analyse the performance  of the proposed schemes  in an uplink cellular network \cite{kumar2014power}.
\appendix
\label{app}
\section*{Proof of Lemma 1}
Given that the typical PU is at distance $r$ from the PBS, the CP averaged over the picocell area is given by
\begin{equation*}
P_{f,c}(T)=\mathbb{E}_r[\mathbb{P}[\text{SIR}>T|r]]=\int_{r>0}^{R}\mathbb{P}[\eta_p(r)>T|r]f_r(r)\text{d}r,
\end{equation*}
\begin{equation*}
=\int_{r>0}^{R}\frac{2r}{R^2} \mathbb{P}\left[\frac{g_p r^{-\alpha}}{ \frac{P_m}{P_p}I_\phi +I_{\psi\setminus P_0}}>T|r\right]\text{d}r,
\end{equation*}
\begin{eqnarray}
&=&\int_{r>0}^{R}\frac{2r}{R^2} \mathbb{P}[g_m>Tr^\alpha(\frac{P_m}{P_p}I_\phi +I_{\psi\setminus P_0})|r]\text{d}r, \nonumber\\
&\stackrel{(a)}=&\int_{r>0}^{R}\frac{2r}{R^2}\mathcal{L_{I_\phi}}(Tr^\alpha \frac{P_m}{P_p})\mathcal{L}_{I_{\psi\setminus P_0}}(Tr^\alpha)\text{d}r,
\label{eq:lemma1}
\end{eqnarray}
 where $ I_{\phi}=\sum\limits_{i\in \phi}{d_i}^{-\alpha}{h_i},$ $ I_{\psi\setminus P_0}=\sum\limits_{j\in (\psi \setminus P_0)}{r_j}^{-\alpha}{g_j}$.  
Here $(a)$ follows from the fact that $g_m \sim \exp(1)$. $\mathcal{L_{I_\phi}}(s)$ and $\mathcal{L}_{I_{\psi\setminus P_0}}(s)$ are the 
Laplace transforms of the random variables $I_\phi$ and $I_\psi$, respectively, evaluated at $s$. Thus,
\begin{eqnarray}
 &&\mathcal{L}_{I_{\psi\setminus P_0}}(s)=\mathbb{E}_{{\psi\setminus P_0},g_j}[\exp(-s\sum_{j \in \psi \setminus P_0} g_j r_j^{-\alpha})] \nonumber\\
&\stackrel{(b)}=&\exp\left(-2\pi\lambda_p\int_{R+R_g^{(p)}-r}^\infty \frac{sx^{-\alpha}}{1+sx^{-\alpha}}x\text{d}x\right )\label{eq:upper}\\  
&\stackrel{(c)}=&\exp\left(-2\pi \lambda_ps^{\frac{2}{\alpha}}\int_{L_p(r)}^\infty \frac{t}{1+t^{\alpha}}\text{d}t\right)\nonumber\\
&\stackrel{(d)}=& \exp\{-\pi s^{2/\alpha}C(L_p(r),\alpha)\lambda_p \} \label{eq:lemma2}
\end{eqnarray}
\begin{equation}
L_p(r)= \frac{R+R_g^{(p)}-r}{r}(T)^{-\frac{1}{\alpha}}, \text { and }C(L_p(r),\alpha)=\frac{{}_2 F_1(1,\frac{\alpha-2}{\alpha},2-\frac{2}{\alpha},-\{L_p(r)\}^\alpha)}{\alpha-2}\{L_p(r)\}^{2-\alpha} .
\label{hyper}
\end{equation}
Here, $(b)$ follows from the fact that $g_j\sim \exp(1)$ and probability generating functional (PGFL) \cite{Stoyan} of the PPP and ${}_2 F_1(a,b,c,z)$ represents the Gauss hypergeometric function. The lower limit of the 
integral in $(b)$ is $R+R_g^{(p)}-r$ because of the fact that all the interfering base stations are at least a distance greater than  $R+R_g^{(p)}-r$. Using a change of variable $t=s^{-\frac{1}{\alpha}}x$, $(c)$ can be followed  and $(d)$ follows after some algebraic manipulation.   In a similar fashion,
\begin{equation}
 \mathcal{L}_{I_\phi}(s)=\exp\{-\pi s^{2/\alpha}C(L_m(r),\alpha)\lambda_m \}.
\label{eq:lemma3}
\end{equation}
where $L_m(r)=\frac{R+R_g^{(m)}-r}{r}\left(T\frac{P_m}{P_p}\right)^{-\frac{1}{\alpha}}$. Using \eqref{eq:lemma2} and \eqref{eq:lemma3} and simplifying \eqref{eq:lemma1}, one obtains
\begin{equation*}
P_{f,c}(T)= \int_{r>0}^{R}\exp(-\pi r^2 T^{2/\alpha}K(r)) \frac{2r}{R^2}\mathrm{d}r,
\end{equation*}
\begin{equation*}
\text{where } K(r)=\left(\frac{P_m}{P_p}\right)^{2/\alpha}C(L_m(r),\alpha)\lambda_m + C(L_p(r),\alpha)\lambda_p.
\end{equation*}
\bibliographystyle{IEEEtran}
\bibliography{bibfile}

\end{document}